\documentclass[10pt, conference, compsocconf]{IEEEtran}
\usepackage{epsfig}
\usepackage{cite}
\usepackage{latexsym}
\usepackage{graphics}
\usepackage{wrapfig}
\usepackage{amsmath,amsfonts,url}
\usepackage{tabularx}
\usepackage{comment}
\usepackage{multirow}
\usepackage{caption}
\usepackage[linesnumbered]{algorithm2e}
\usepackage{color}
\usepackage{bm}
\usepackage{setspace}
\usepackage{xspace}
\usepackage{amsthm}

\usepackage{subfig}

\usepackage{tabularx,ragged2e,booktabs,caption}

\usepackage[linesnumbered]{algorithm2e}
\SetAlCapSkip{1em}

\SetKw{Parallel}{parallel}

\def\bits{\{0,1\}}
\newcommand{\xor}{{\oplus}}

\newcommand{\heading}[1]{\vspace{5pt}\noindent\underline{\textsc{#1}}}

\newcommand{\noskipheading}[1]{\noindent\underline{\textsc{#1}}}

\newcommand{\Keys}{\mathrm{Keys}}
\newcommand{\Dom}{\mathrm{Dom}}

\newcommand{\secref}[1]{Section~\ref{#1}}

\newcommand{\figref}[1]{\textmd{Fig.}~\ref{#1}}

\renewcommand{\eqref}[1]{\mbox{Eq.~(\ref{#1})}}

\newcommand{\concat}{{\,\|\,}}

\newcommand{\vecM}{\bm{M}}

\newcommand{\flag}{\textsl{last}}
\newcommand{\IF}{\textbf{if }}
\newcommand{\ELSE}{\textbf{else }}
\newcommand{\THEN}{\textbf{then }}

\newcommand{\Header}{\textsl{Header}}

\newcommand{\sysrm}{\textrm{CryptMPI}}

\newtheorem{proposition}{Proposition}

\newcommand{\Enc}{\texttt{Enc}}
\newcommand{\Dec}{\texttt{Dec}}
\newcommand{\Gen}{\texttt{Gen}}
\newcommand{\calM}{{\mathcal M}}

\newcommand{\Tcomm}{T_{\mathrm{comm}}}
\newcommand{\acomm}{\alpha_{\mathrm{comm}}}
\newcommand{\bcomm}{\beta_{\mathrm{comm}}}
\newcommand{\Ted}{T_{\mathrm{enc}}}
\newcommand{\aed}{\alpha_{\mathrm{enc}}}

\newcommand{\vecC}{\boldsymbol{C}}
\newcommand{\AES}{\mathrm{AES}}

\newcommand{\IRecv}{\texttt{MPI\_IRecv} }
\newcommand{\Recv}{\texttt{MPI\_Recv} }
\newcommand{\ISend}{\texttt{MPI\_ISend} }
\newcommand{\Send}{\texttt{MPI\_Send} }

\newcommand{\Wait}{\texttt{MPI\_Wait} }
\newcommand{\Waitall}{\texttt{MPI\_Waitall} }

\newcommand{\Init}{\texttt{MPI\_Init}}
\newcommand{\pk}{{pk}}
\newcommand{\sk}{{sk}}
\newcommand{\Gather}{\texttt{MPI\_Gather}}
\newcommand{\Scatter}{\texttt{MPI\_Scatter}}
\newcommand{\calK}{\mathcal{K}}
\newcommand{\Unencrypted}{{\emph{Unencrypted}}}
\newcommand{\Naive}{{\emph{Naive}}}
\newcommand{\Ti}{T_i}
\newcommand{\Tt}{T_c}
\newcommand{\Te}{T_e}

\title{\sysrm: A Fast Encrypted MPI Library}
\author{
  
\IEEEauthorblockN{Abu Naser \hspace{0.1in} Cong Wu \hspace{0.1in} Mehran Sadeghi Lahijani \hspace{0.1in} Mohsen Gavahi \hspace{0.1in} Viet Tung Hoang \hspace{0.1in} Zhi Wang \hspace{0.1in} Xin Yuan}
\IEEEauthorblockA{Department of Computer Science, Florida State University, Tallahassee, FL 30306\\
\{naser,wu,sadeghil,gavahi,tvhoang,zwang,xyuan\}@cs.fsu.edu}
}

\date{}

\begin{document}

\maketitle

\begin{abstract}
  The cloud infrastructure must provide security 
	for  High Performance Computing (HPC) applications 
	of 
	sensitive data to execute in such an environment. 
	However, supporting security 
	in the communication infrastructure of today's public cloud
  is challenging, because
	current networks for data centers are so fast that adding encryption can incur
  very significant overheads. In this work, we introduce $\sysrm$,
  a high performance 
	encrypted MPI library that 
	supports communication  with both integrity and privacy. 
  We present the techniques in $\sysrm$ and report our benchmarking results
  using micro-benchmarks and NAS parallel benchmarks. The evaluation results
  indicate that the aforementioned techniques are effective in improving the
  performance of 
	encrypted communication.

\end{abstract}

  
\section{Introduction}

High Performance Computing (HPC) applications are often used to process
sensitive data, such as medical, financial, or scientific documents. 
In the past, HPC jobs were run in trusted local clusters of no external access, 
and security was taken for granted.
But the current trend of moving everything to the cloud has upended this assumption, 
introducing unprecedented security challenges to HPC. 
This paper aims to partially address this problem 
in the context of Message Passing Interface (MPI) environment, which is the de facto library for 
message-passing applications.

To ensure privacy and integrity of sensitive data, 
an obvious solution is to encrypt MPI communication. 
The easiest way to do that is to rely on existing encryption mechanisms at lower levels, 
such as IPSec. 
This approach allows one to encrypt MPI communication without modifying the MPI library. 
The second option is to naively incorporate  encryption to the MPI library 
by encrypting the whole message at the sender side, transmitting
the entire ciphertext, and decrypting the ciphertext at the
receiving side~\cite{Cluster:Naser19}.

\figref{fig:ipsec} illustrates the performance of the unencrypted baseline,
IPSec, and our library CryptMPI on a 10 Gbps Ethernet network. 
Although IPSec and CryptMPI rely on the same underlying encryption scheme, 
the performance of the former is far worse. 
In particular, the throughput that can be achieved by IPSec is only a small
fraction of the raw network throughput. In
our system, IPSec only achieves about one-third of the network throughput
for 1MB messages. Moreover, IPSec sequentializes encrypted communication
when there are multiple concurrent flows in the system. As shown
in \figref{fig:ipsec}, the aggregate throughput remains the same
when the number of concurrent flows increases from one to~four.

\begin{figure}[t]	
    \centering
      \includegraphics[width=0.4\textwidth]{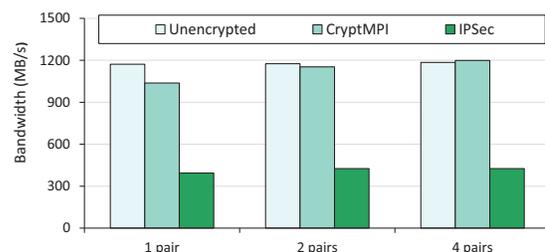}%
    \vspace{-1.5ex}
    \captionsetup{singlelinecheck=false}
    \caption{Motivating example: IPSec throughput on 10Gbps Ethernet
      with multiple concurrent flows and 1MB messages.}
    \label{fig:ipsec}
    \vspace{-1.5ex}
    \end{figure}

Several prior works~\cite{Cluster:Naser19,Ruan:2012:EMP:2197079.2197242,Maffina2013,RTSJ16,Shivaramakrishnan2014}
instead opt for naively adding encryption to the MPI library. 
But except for the recent work of Naser et al.~\cite{Cluster:Naser19}, 
all other prior systems contain severe security flaws. 
For example, ES-MPICH2~\cite{Ruan:2012:EMP:2197079.2197242}, the first encrypted MPI library,
uses the weak ECB (Electronic Codebook) mode of operation that has known
vulnerabilities~\cite[page 89]{Book:KL14}.

Even if one correctly implements the naive approach of adding encryption to the MPI library
as in~\cite{Cluster:Naser19}, this route is doomed to give very unsatisfactory performance, 
due to recent advances in the networking technology for data centers. 
\figref{fig:motivate} compares the communication performance of
unencrypted MPI library (MVAPICH) with the naive method
using the AES-GCM encryption scheme with a 128-bit key on a cluster with 40Gbps InfiniBand
in the Ping-Pong setting. 
As shown in \figref{fig:motivate},  naively adding encryption can incur very high overheads. 
For example, when the message size is 1MB, the throughput
drops from 3.0 GBps for unencrypted communication
to 1.2GBps for encrypted one.

\begin{figure}[t]
\centering
\vspace{-2.5ex}
\includegraphics[width=0.4\textwidth]{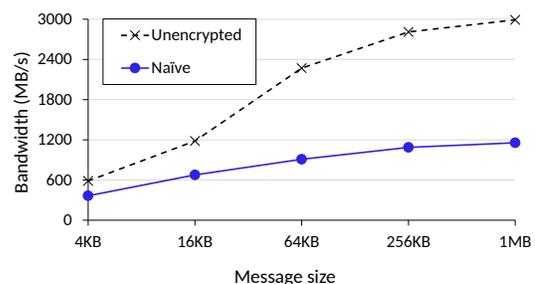}
\vspace{-1.5ex}
\caption{Motivating example: the one-way throughput by naively incorporating 
AES-GCM with a 128-bit key into MVAPICH.}
\label{fig:motivate}
\vspace{-3.0ex}
\end{figure}

The results above show that in today's data centers, even with the
state-of-the-art processor hardware support for cryptographic operations
(e.g., Intel AES-NI instructions to accelerate the AES algorithm,
or the x86 CLMUL instruction set to improve the speed of finite-field
multiplications), the overheads of encryption 
can easily dominate the encrypted communication time. This calls
for a careful design of an encrypted MPI library so as to exploit the
computing power at end-nodes and to minimize
the encryption  overheads. Compute nodes in modern systems
usually have lots of CPU cores with Hyper-Threading
Technology (HTT). 
For instance, each node of the world's fastest supercomputer Summit
supports 176 hyper-threads~\cite{summit},
meanwhile a single processor~\cite{amdryzen} of AMD has
128 hyper-threads. These imply that in the very near future compute nodes will
be equipped with more hyper-threads than what applications can use. 
$\sysrm$ targets this type of systems by leveraging hyper-threads to deliver
higher performance.

In this work, we introduce $\sysrm$, a
high-performance encrypted MPI library. $\sysrm$ encrypts MPI communication under AES-GCM, 
providing both integrity and privacy. Since encryption overheads are relatively small for small messages~\cite{Cluster:Naser19}, 
this work focuses on optimizing the point-to-point communication performance for large messages.
$\sysrm$ employs a technique to maintain the required security properties while allowing the
communication and encryption to be pipelined and be performed by multiple threads. Moreover, 
by using an accurate performance model for Ping Pong experiments, 
$\sysrm$ automatically derives the parameters used at runtime for each communication, such
as the number of threads 
and the size of each chunk of data in 
pipelined encrypted communication.
We will describe the techniques in $\sysrm$ and report our benchmarking
results, which demonstrates that our methods are effective in achieving fast encrypted
MPI communication. 

The rest of the paper is structured as follows. Section~\ref{sec:related}
discusses the related work. Section~\ref{sec:backg}
introduces the background of this work.
Section~\ref{sec:oursystem} presents our system. Section~\ref{sec:perf}
reports
the results of our performance study. Finally, Section~\ref{sec:conclu}
concludes the paper.

\section{Related work}
\label{sec:related}

Most prior systems of encrypted
MPI~\cite{Ruan:2012:EMP:2197079.2197242,Maffina2013,RTSJ16,Shivaramakrishnan2014} only use weak encryption schemes, and therefore fail to deliver privacy. 
For example, the ES-MPICH2 system~\cite{Ruan:2012:EMP:2197079.2197242}
relies on the Electronic Codebook (ECB) mode, 
which is deemed insecure by cryptography textbooks~\cite[page 89]{Book:KL14}. 
Likewise, the VAN-MPICH2 system~\cite{RTSJ16} attempts to build its
own encryption scheme, but this is trivially broken~\cite{Cluster:Naser19}. 
In addition, those systems  also do not provide integrity. 
Some, for example, use encryption with a form of checksum~\cite{Maffina2013}, 
but this approach does not provide integrity if the encryption scheme is the standard Cipher Block Chaining (CBC) mode of encryption~\cite{C:HB01}. 
See the work of Naser et al.~\cite{Cluster:Naser19} for details of
the insecurity of the prior systems of encrypted MPI. Additionally, these
earlier systems were designed when the communication bottleneck was in
network links. As such, they do not consider techniques to optimize the use
of computing resources at end-nodes to maximize the encrypted communication
performance. 

The work of Naser et al.~\cite{Cluster:Naser19} is the only existing
system that provides adequate privacy and integrity protection for MPI. 
However, they only aim to investigate the performance of existing
cryptographic libraries in the MPI environment via a vanilla implementation. 
Consequently, their system incurs high overheads as shown in
\figref{fig:motivate}. This motivates us to build a system that
meets both security and efficiency requirements of MPI applications.


\section{Background}
\label{sec:backg}

\noskipheading{Blockciphers.} A blockcipher on message space $\calM$ and \allowbreak
keyspace~$\calK$ is a function
$E: \calK \times \calM \to \calM$ such that for every key $K \in \calK$, the function $E_K(\cdot) = E(K, \cdot)$ is 
a permutation on~$\calM$. 
A typical blockcipher is AES.

\heading{Symmetric Encryption.}  Syntactically, a (symmetric) encryption scheme is a triple of algorithms
$(\Gen, \Enc, \Dec)$. Initially, the sender and receiver need to share a secret key~$K$
that is generated by~$\Gen$. 
Each time the sender wants to encrypt a message $M$, she creates a ciphertext $C$ via
$C \gets \Enc(K, M)$. 
Upon receiving~$C$, the receiver decrypts via $M \gets \Dec(K, C)$. 
An encryption scheme is often built upon a blockcipher such as~AES.

Traditionally, an encryption scheme's only goal is to deliver \emph{privacy}, 
meaning that an adversary cannot distinguish the ciphertexts of its chosen messages
with those of uniformly random strings of the same length. 
Several standard encryption methods, such as 
the Cipher Block Chaining (CBC) mode or Counter (CTR) mode~\cite{NIST:38A}
fail to provide  \emph{integrity} (meaning that the adversary cannot modify ciphertexts without being detected).\footnote{Actually, 
an adversary may still replace a ciphertext by an old one; this is known as a \emph{replay attack}. 
We do not consider replay attacks in this work. 
}
In this work, we focus on the Galois Counter Mode (GCM)~\cite{NIST:38D}. 
GCM is the fastest standard encryption method that \emph{provably} provides both privacy and integrity~\cite{McGrew:2004:SPG:2158177.2158215}.

\smallskip

\noskipheading{GCM overview.} GCM is a \emph{nonce-based} encryption scheme. 
Each time the sender encrypts, she needs to provide a 12-byte \emph{nonce}~$N$, a string that should never repeat, 
to produce a ciphertext $C \gets \Enc(K, N, M)$. 
Conventionally, a GCM nonce is implemented as a counter or a uniformly random string. 
Each ciphertext $C$ is 16-byte longer than the corresponding message $M$
as it includes a 16-byte tag for integrity checking. 
The sender then needs to transmit both the ciphertext $C$ and the nonce $N$ to the receiver. 
The latter then decrypts via $M \gets \Dec(K, N, C)$. 
Thanks to its speed and provable security, GCM is widely used in network protocols, 
such as SSH, IPSec, and TLS.

According to NIST standard SP 800-38D~\cite{NIST:38D}, 
the blockcipher of GCM must be AES; we  often write AES-GCM
to refer to an instantiation of GCM with AES as the blockcipher. 
There are three corresponding key sizes for AES: 128, 192, or 256 bits. 
Longer key length means better security against brute-force attacks, 
at the cost of slower execution time. 
In this work we only consider 128-bit keys to achieve the best possible performance.

\heading{Public-key Encryption.} Public-key encryption (PKE) is analogous to symmetric-key encryption, 
but the (public) encryption key is different from the (secret) decryption key. 
Standard PKE scheme (such as RSA~\cite{RSA}) are more expensive than standard symmetric-encryption schemes (such as GCM)
by several orders of magnitude. 
One therefore often just uses PKE to distribute, say a shared GCM key, 
and then switches to GCM encryption. 
In this work, we use the OAEP method of RSA~\cite{OAEP}, 
which is \emph{provably} secure.

\section{\sysrm: A Fast Encrypted MPI Library}
\label{sec:oursystem}

$\sysrm$ supports encrypted MPI communication via AES-GCM from the BoringSSL library~\cite{boringssl}. 
It is implemented on top of MVAPICH2-2.3.2 for systems with
InfiniBand and Omni-path-based interconnects. 

Our starting point is the vanilla implementation of Naser et al.~\cite{Cluster:Naser19}.
Since encryption overheads are relatively low
for point-to-point communication of small messages \cite{Cluster:Naser19},
this work focuses on optimizing point-to-point communication of
large messages (meaning 64KB and beyond). 
The following MPI routines are modified: $\Send$, $\Recv$, $\ISend$, $\IRecv$, $\Wait$, and $\Waitall$.

To achieve high performance, 
$\sysrm$ incorporates two optimization techniques: 
pipelining and multi-threading. 
Integrating those with encryption without compromising security is challenging. 
We resolved that by using theory from streaming encryption~\cite{C:HRRV15} and Google's Tink library~\cite{tink}. 
It is also nontrivial to decide the values of the parameters of those optimizations---for example, the number of threads in multi-threading---for 
maximizing the performance. 
To find the best choice of the parameters, we developed a
performance model for 
the parameterized version of CryptMPI in the ping-pong setting. 
CryptMPI then uses that model to derive the values for the parameters that
maximize the ping-pong performance for each communication. 
Empirical results confirm that these are also good choices in
other settings. 

Below, we will first describe our threat model, 
and then discuss the details of the optimization techniques. 
For the convenience of discussion, 
we will adopt the following notation. 
Let $\bits^n$ be the set of $n$-bit binary strings, 
and let $x \concat y$
denote the concatenation of the two strings $x$ and~$y$.
For an integer $r$, we write $[x]_r$ to denote a $r$-byte encoding
of a string~$x$.

\heading{Threat model.}
Under our model, we consider an external 
adversary that can observe and tamper with network packets. 
We however assume that compute nodes are secure. 
As a result, $\sysrm$ only encrypts inter-node communication.

\heading{Pipelining.}
A reason for the high encryption overhead in the vanilla
implementation of Naser et al.~\cite{Cluster:Naser19}
is that the receiver has to wait for the sender to
encrypt the \emph{entire} message. A simple solution is to chop the
message into multiple segments, and encrypt and transmit them individually, which
results in pipelining the tasks in  the sender and receiver.
Under this approach, the receiver can start decrypting as soon as the
first ciphertext segment arrives. Hence, with pipelining, message transmission,
encryption, and decryption can all be overlapped.

The pipelining method above however breaks the authenticity of GCM:
the adversary can reorder the ciphertext segments of a message,
or even drop some of them without detection. This is a known issue for
streaming authenticated encryption~\cite{C:HRRV15}.
The treatment, suggested by~\cite{C:HRRV15},  is to embed the following
information in the nonce:
(i) a flag to indicate whether the segment is the last one
for the given message, and (ii) a counter to specify the position of the current segment within the given message.
Typically one uses one byte to encode the flag, and four bytes for the counter.
This approach is, however, problematic in our context.
In particular, if we pick the seven-byte remainder of a GCM nonce as a random string,
then we will be likely to have nonce repetition after encrypting just $2^{28}$ messages,
due to the well-known Birthday Problem.\footnote{
The Birthday Theorem (see, for example, \cite[Appendix A.4]{Book:KL14}) states that if we sample $q$ elements $x_1, \ldots, x_q$ uniformly and independently from a set of size $N$,
for $q \leq \sqrt{2N}$,
then the chance that there are some $i \ne j$ such that $x_i = x_j$ is at least $q(q - 1) / 4N$.}

To resolve the issue above, we adopt the approach in Google's 
library Tink~\cite{tink} to use a different \emph{subkey} for each message, but our key derivation method is simpler. 
Specifically, to encrypt a message $\vecM$ under a master key~$K$, 
we first pick a 128-bit random seed $V$, and derive a subkey $L \gets \mathrm{AES}_K(V)$. 
We then use GCM with key $L$ to encrypt the segments of~$\vecM$, 
with the nonces following the suggestion in~\cite{C:HRRV15}. 
The seed $V$ and the message length $|\vecM|$ will be sent together with the first ciphertext segment. 

The pseudocode of our encryption method  is given in Algorithm~\ref{algo:chop}. 
In addition to the ciphertext segments, the ciphertext~$\vecC$ consists of a header
that provides the seed $V$, the length of the message $\vecM$, and the size of a message segment.
In implementation, the computation of $\vecC = (\Header, C_1, \ldots, C_t)$
will be interleaved with non-blocking communication
so that the receiver can start decrypting $C_i$ while the sender is still encrypting or sending~$C_{i + 1}$.

\RestyleAlgo{boxed}
\IncMargin{1em}
\begin{algorithm}[h]
  \DontPrintSemicolon
  \SetKwFunction{alltoall}{MPI\_Alltoall}\SetKwFunction{nonce}{RAND\_bytes}
  \BlankLine
	\textbf{Input:} A message $\vecM$ of $m$ bytes. \;
	\textbf{Parameter:} A key $K$, and an integer $t$  \; 
Pick a 16-byte random seed $V$ \;
\tcc{Derive a subkey $L$} 
$L \gets \AES_K(V)$ \\
$s \gets \lceil m / t \rceil$;~ $\Header \gets (V, m, s)$ \;  
\tcc{Chop $\vecM$ into $t$ segments} 
  $(M_1, \ldots, M_t) \gets \vecM$ \;
\tcc{Encrypt each segment} 
\For{$i\gets 1$ \KwTo $t$}{	
\IF $i = m$ \THEN $\flag \gets 1$ \ELSE $\flag \gets 0$   \;
$N_i \gets [0]_7 \concat [\flag]_1 \concat [i]_4$;~ $C_i \gets \Enc(L, N_i, M_i)$ \;
}
\Return $(\Header, C_1, C_2\ldots, C_t)$
\caption{How to chop and encrypt a message~$\vecM$ under master key $K$.  
Here 
$\Enc(L, N, M)$ denotes the use of GCM to encrypt a message $M$ under nonce $N$ and key $L$.
 }\label{algo:chop}

\end{algorithm}\DecMargin{1em}

When the receiver gets the header, it 
derives the number of segments $t$ that it needs to decrypt from the segment size~$s$
and the message size $m$. 
Later, if GCM decryption flags some ciphertext segment as invalid, 
or if the receiver does not get the correct number of ciphertext segments, 
it will report a decryption failure. 
In particular, if an adversary tampers with the header then thanks to the integrity property of GCM, 
this will result in a decryption failure.

A very recent paper by Hoang and Shen~\cite{CCS:HS20}
gives a proof for the security of  a generalized version of the encryption in Google's library Tink
that also covers our approach. 
Under the analysis in~\cite{CCS:HS20}, our encryption scheme provides both privacy and integrity as long as
the 128-bit seeds are distinct, 
which is very likely according to the following result. 

\begin{proposition}
Suppose that we sample $q$ seeds $V_1, \ldots, V_q$ uniformly at random from $\bits^{128}$. 
Then the chance that they are distinct is at least $1 - q^2 / 2^{129}$.
\end{proposition}
\begin{proof}
Since the seeds are sampled independently and uniformly random, 
for every $i \ne j$, the chance that $V_i = V_j$ is at most $2^{-128}$. 
Since there are ${q \choose 2} \leq q^2 /2$ pairs of seeds, 
by the Union Bound, the chance that two seeds in some pair are the same is at most $q^2 / 2^{129}$. 
\end{proof}

\heading{Multi-threading encryption/decryption.} As shown
in \figref{fig:motivate},
there is a huge gap between single-core encryption speed and 
MPI bandwidth. 
Data in \figref{fig:motivate} are for a 40Gbps interconnect; 
the situation is much worse for current generation HPC systems with 100Gbps interconnects. 
On the other hand, modern machines may have many cores that MPI applications
do not fully utilize. A natural way to fill the gap is to 
chop a message into $t$ segments, 
and use $t$ threads to encrypt them simultaneously via Algorithm~\ref{algo:chop}, 
where $t$ is a parameter. 
The header will be computed and sent first.
Then, each thread $i$ will compute the corresponding ciphertext segment~$C_i$,
and the entire $(C_1, \ldots, C_t)$ will be sent as a whole.

\heading{Putting things together: $(k, t)$-chopping algorithm.}  To combine pipelining and multi-threading,
we chop a message into $kt$ segments and encrypt them via Algorithm~\ref{algo:chop},
where $k$ and $t$ are parameters, 
but the interleaving of encryption and communication is as follows. 
Again, the header will be computed and sent first.
Then, we will use $t$ threads to compute the first $t$-segment chunk $(C_1, \ldots, C_t)$ and send it as a whole,
and then compute  and send the next chunk $(C_{t + 1}, \ldots, C_{2t})$, and so on.
For convenience, we refer to this algorithm as \emph{$(k, t)$-chopping}. 
Note that when $k = 1$ then the
$(k, t)$-chopping algorithm falls back to basic multi-threading encryption (namely using
$t$ threads to encrypt the message), 
and when $t = 1$, it corresponds to the special case where only pipelining is used.

Still, GCM encryption is somewhat slow for tiny messages.
The encryption speed, however, gathers momentum quickly and gets saturated at around 32 KB.
CryptMPI, therefore,
uses the $(k, t)$-chopping algorithm only if the message size is at least 64KB.
For small messages, we will directly encrypt via GCM so that we do not waste time in deriving subkeys.
To simplify the decryption process, we will still create a header when encrypting a small message;
here a 12-byte nonce will be included instead of a 16-byte random seed.
Headers will additionally contain an opcode to inform receivers of the encryption algorithm.

We stress that one needs to maintain \emph{two} separate keys, one for Algorithm~\ref{algo:chop} to encrypt large messages,
and another for direct GCM encryption of small messages.
Violating this key separation will break security.
To see why, suppose that we use a single key $K$ for encrypting messages of all sizes,
and suppose that at some point we encrypt a known 16-byte message $X$ directly via GCM under a random nonce~$N$.
The resulting ciphertext consists of $\AES_K(V) \xor X$ and a tag~$T$,
where $\xor$ denotes the xor of two equal-length strings, and $V = N \concat [1]_{4}$.
Since $X$ is known, an adversary can extract $L \gets \AES_K(V)$ from the ciphertext.
It then can create a valid ciphertext for \emph{any} large message $\vecM$ as follows.
Instead of picking a 16-byte random seed and deriving the corresponding subkey,
it uses the string $V$ above as the seed and the string~$L$ above as the subkey.
It then runs lines 5--11 of Algorithm~\ref{algo:chop} to encrypt $\vecM$.

\heading{Modeling $(k, t)$-chopping.}
To find the best values of $k$ and~$t$ in the $(k, t)$ chopping algorithm, 
we developed a performance model for it in the ping-pong setting, 
for generic $k$ and~$t$. The model is highly accurate on all of the systems
that we tested including all of the ones used in our experiments.
\figref{fig:model-valid} shows an example: the predicted results and measured
results using the ping-pong program in our Noleland cluster with InfiniBand
(described in Section~\ref{sec:perf}) match well.   
We later used the model to estimate the optimal choice of $k$ and~$t$ that maximizes the 
ping-pong performance. 
The model has two sub-components: a model for
unencrypted communication, and another for multi-threading encryption.

\begin{figure}[t]
	\centering
	\vspace{-2.5ex}
		\includegraphics[width=0.35\textwidth]{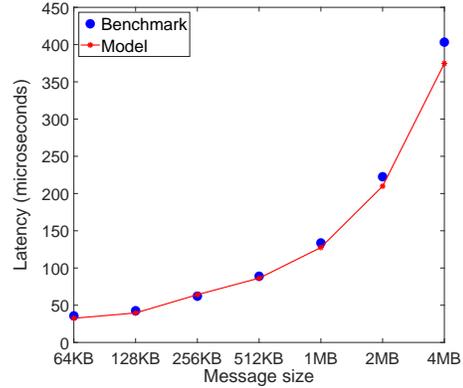}
		\vspace{-1ex}
	\caption{Latency of encrypted ping-pong under CryptMPI on InfiniBand: benchmark versus model prediction. }
		\vspace{-0.5ex}
	\label{fig:model-valid}
	\vspace{-2.5ex}
\end{figure}

We use the classic Hockney model~\cite{Hockney94} to model
communication. Under the Hockney model, the time to send an $m$-byte message is
\[
\Tcomm(m) = \acomm + \bcomm \cdot m \enspace,
\]
where $\acomm$ is the network latency, and $\bcomm$ is the transmission rate.
The parameters $(\acomm,\bcomm)$ can be derived from a ping-pong benchmark.

To model the time $\Ted(m, t)$ to encrypt an $m$-byte message via~$t$ threads, we picture multi-threading encryption as multiple-pair
communication. Thanks to this viewpoint, 
we can use the max-rate model of Gropp, Olson, and Samfass~\cite{gropp2016modeling} for concurrent point-to-point communications
for modeling $\Ted(m, t)$. Specifically, 
\[
\Ted(m, t) = \aed + \dfrac{m}{A+B\cdot (t-1)} \enspace,
\]
where $\aed$ is the initial overhead, $A$ is the encryption throughput for the first thread, 
and $B$ is the encryption throughput for each subsequent thread. 
This captures the fact that adding more threads after the first one will have degraded performance per thread.


For given $k$ and $t$, we now consider modeling the
$(k, t)$-chopping algorithm in the ping-pong setting of $m$-byte messages. 
For simplicity, assume that $m$ is a multiple
of $kt$, and let $s = m / k$ be the size of a $t$-segment chunk. 
We assume that the encryption
and decryption of a message take a similar amount of time, which holds for AES-GCM.

Recall that under the $(k, t)$-chopping algorithm, 
at the beginning, the sender encrypts the first  $t$-segment chunk with $t$ threads,
which takes $\Ted(s, t)$ time units. Next, encrypting the $i$-th plaintext chunk
and transmitting the $(i-1)$-th  ciphertext one overlaps with each another,
for $i = 2, \ldots, k$, 
and thus those take totally of  
\[ (k - 1)  \cdot \max\{ \Ted(s, t), \bcomm \cdot s \}  \] 
time units. Note that as the ciphertext chunk are sent successively,
the latency term  $\acomm$ is only manifest in the transmission of the
last chunk. 
It then takes $\Tcomm(s)$ additional time units for the
last ciphertext  chunk to reach the receiver. 
During this time, the receiver can decrypt all but the last ciphertext chunk, 
and thus it will need $\Ted(s, t)$ additional time units to decrypt the
last ciphertext chunk.

\smallskip 

\noindent
\textbf{The complete model:} Summing up, the total time for the
$(k, t)$-chopping algorithm to send an $m$-byte message is 
\[ 
2 \cdot \Ted(s, t) +  (k-1)\cdot \max \{\Ted(s, t), \bcomm\cdot s \}+ \Tcomm(s) \enspace.
\] 
If the network speed is slower than the (multi-threading) encryption speed, 
then the total time is around $\Tcomm(m) + 2 \cdot \Ted(s, t)$, 
meaning that we have to pay for the encryption-decryption cost of
a single chunk instead of the entire message, 
and thus the encryption and decryption cost almost vanishes. 
When the network speed is faster than the encryption speed, 
the total time is nearly $(k + 1) \cdot \Ted(s, t) + \Tcomm(s)$, 
meaning that the encryption-decryption cost is reduced to
nearly a factor of $2t$
in comparison to the naive method when single-core encryption, transmission, and single-core decryption
are performed in sequence. 

There are five parameters in our model: $\acomm$, $\bcomm$, $\aed$, $A$, and $B$.
For a given system, the parameters can be derived by the measurement
results from the (often publicly available) ping-pong benchmark and the multi-threading encryption benchmark, together with some library and architecture information.

In the following, we will use our local Noleland cluster with InfiniBand as an
example to illustrate the process to derive the model parameters. The measured
unencrypted ping-pong performance for the cluster
is shown in \figref{fig:pingpong_infiniband}. These data, together with the
library information (eager threshold), are used to derive the
parameters for the Hockney model via least squares regression,
which is shown in Table~\ref{tab:hockney_param}.

\begin{table}[!t]
		\centering
		\captionof{table}{The values of parameters $\acomm$  and $\bcomm$ for unencrypted one-to-one communication on InfiniBand.}
			\begin{tabular}{p{0.12\linewidth}*4{p{0.2\linewidth}}}
			\toprule[1.25pt]
			&  & {$\acomm$ ($\mu$s) } & { $\bcomm$($\mu$s/B)} \\ \midrule  
			\multirow{1}{*}{\textbf{Eager}} & 
			& 5.54 & $7.29 \times 10^{-5}$  \\\midrule  
			\multirow{1}{*}{\textbf{Rendezvous}} 
			&  & 5.75 & $7.86\times 10^{-5}$ \\ 
			\bottomrule[1.25pt]
		\end{tabular}
		\label{tab:hockney_param} 
		\vspace{-1ex}
\end{table}

\begin{table}[!t]
		\centering
		\captionof{table}{The values of parameters $(\aed, A, B)$ for multi-threading encryption.}
		\begin{tabular}{cccc}
			\toprule[1.25pt]
			\textbf{} &  {$\aed$ ($\mu$s) } & { $A$ (B/$\mu$s)}  & $B$ (B/$\mu$s) \\ \hline
			\textbf{Small} & 4.278&  5265 & 843 \\ \midrule  
			\textbf{Moderate} & 4.643 &  6072 & 4106 \\    \midrule  
			\textbf{Large} & 5.07 &  5893 & 5769 \\   
			\bottomrule[1.25pt]
		\end{tabular}
		\label{tab:enc_param} 
	\vspace{-1ex}
\end{table}

The measured multi-threading encryption throughput results are shown
in \figref{fig:enc_dec_gcc_throughput}. These results, together with
architecture information (L1 and L2 cache sizes), are used to derive
the model parameters. To account for the L1 and L2 cache size in our system,
we consider three levels of message size: small (below $32$KB),
moderate (from 32KB to under 1MB), and large (at least $1$MB).
We then use different values of the parameters $(\aed, A, B)$
for each level.
The values of $(\aed, A, B)$ are given in Table~\ref{tab:enc_param};
they are obtained via Matlab's non-linear least square on the result of the
encryption throughput benchmark.

\heading{Parameter selection.}
Based on the model above, for each message size, 
we estimate the values for the parameters $k$ and~$t$ to
maximize performance for the Ping Pong setting. 
Empirical results confirm that those values
are good choices  for other settings. 
In other words, although we set the parameters for every setting based on the model of Ping Pong, 
those choices turn out to work well. 
Specifically, for a message of $m$-KB, we let 
\[k = \lfloor \max\{ 1, m / 512\} \rfloor \enspace. \]
Our experiments were performed in two systems: our local Noleland cluster
and an Xcede resource,  PSC Bridges cluster. 
For the Noleland cluster with InfiniBand, we let 
\[
t =
\begin{cases}
2 & \text{if } 64 \leq m < 128\\
4 & \text{if } 128 \leq m < 512\\
8 & \text{if } m \geq 512 \enspace.
\end{cases}
\]
In PSC Bridges with Omni-path, we set
\[
t =
\begin{cases}
4 & \text{if } 64 \leq m < 256\\
8 & \text{if } 256 \leq m < 512\\
16 & \text{if } m \geq 512 \enspace.
\end{cases}
\]

$\sysrm$ also takes into account some system constraints to
avoid creating too many threads or too many outstanding
send requests due to pipelining. Both can be detrimental to performance. 
The conceptual idea applies to all systems, but specific numbers may differ. 
Let $T_0$ be the number of hyper-threads that the current rank is allocated.\footnote{
In Noleland there are totally $T = 32$ hyper-threads, whereas $T = 28$ for PSC Bridges. 
If there are $r$ ranks in a node then each rank will be allocated $T_0 = \lfloor T / r \rfloor$ hyper-threads. 
}
Typically~$T_1$ hyper-threads out of the $T_0$ available ones are used for communication; in our systems $T_1 = 2$. 
Requesting more threads than the limit $T_0 - T_1$ will generally slow down
performance, so $T_0 - T_1$ serves as an upper bound of the number of threads
that we run. Thus, CryptMPI will request for $\min\{T_0-T_1, t\}$ threads. 
If there are more than 64 outstanding send requests in this MPI rank
then CryptMPI will set $k = 1$. 
Otherwise, $k$ is set as suggested by the model. 

	The pipelining may cause problems if it is not handled carefully.
	For instance, a receiver may post a bigger message size request instead of the
	actual message size during the receive.  To solve the issue, we post one or multiple pipeline
	requests based on posted message size in the receiver. However, that may break the order
	of the pipeline or could block the progress of other pipeline chunks. To avoid those
	issues we post each pipelined request with a different tag in both the sender and receiver.
	Recall that receiver receives a header first for each message. From the header,
	we extract the actual message size and canceled any extra pipeline requests. 

%

\heading{Key distribution.} We modified the $\Init$ routine to include the following basic key distribution. 
At first, each process $i$ generates a pair $(\pk_i, \sk_i)$ of RSA public and secret keys. 
It then participates in the unencrypted $\Gather$ to collect all public keys at the process $0$. 
Next, process $0$ generates two AES keys $(K_1, K_2)$. 
It then uses $\pk_i$ to encrypt $(K_1, K_2)$ for each process $i$ via  the RSA-OAEP method~\cite{OAEP} of the BoringSSL library, 
resulting in a corresponding ciphertext~$C_i$. 
It then runs $\Scatter$ to distribute each $C_i$ to process~$i$. 
The latter then uses $\sk_i$ to decrypt $C_i$, getting $(K_1, K_2)$.

The key distribution above works against a passive adversary that merely collects traffic packets, 
thanks to the (provable) privacy property of RSA-OAEP~\cite{OAEP}.
It, however, fails to protect against active adversaries. 
Such an adversary can create its own pair $(K^*_1, K^*_2)$, 
encrypts that with each (publicly available) $\pk_i$ to get a corresponding $C^*_i$, intercepts and replaces the packets
of $C_i$ by those of $C^*_i$. 
MPI nodes cannot detect this modification, and will subsequently communicate under the adversary's keys. 
Dealing with active adversaries requires a proper public-key infrastructure; 
we leave that as future work.

\section{Experiments} \label{sec:perf}

In this section, we empirically evaluate the performance of $\sysrm$ on
two platforms and compare $\sysrm$ with
conventional unencrypted MPI, which will be called \Unencrypted,
and the naive approach of
Naser et al.~\cite{Cluster:Naser19}, which will be called {\Naive}.
We begin with a description of our experiment setup
and our benchmark methodology. 
We then report the experiment results.

\heading{System setup.}
The experiments are performed on two systems. One is a local {\em Noleland}
cluster with InfiniBand connectivity. Machines in our cluster are
Intel Xeon Gold 6130 with 2.10 GHz base frequency. Each node 
has 16 cores, 32 threads, and 187GB DDR4-2666 RAM, and runs CentOS~7.5. 
It is equipped with a 100 Gigabit InfiniBand one (Mellanox MT28908 Family, ConnectX-6). 
Allocated nodes were chosen manually. Each experiment used the same
node allocation for all measurements.  All ping-pong results used two
processes on different nodes.

The other system is an Xsede resource, PSC Bridges with Regular Memory
partition~\cite{XSEDE}. The CPUs
are Intel Haswell E5-2695 v3, with 2.3--3.3 GHz base frequency. Each
node has 14 cores, 28 threads, and 128GB DDR4-2133 RAM, and runs CentOS~7.6.
It is equipped with a 100 Gigabit Intel Omni-path interconnect.

We implemented our prototypes on top of 
MVAPICH2-2.3.2  for both InfiniBand and Omni-path. The underlying
encryption library is BoringSSL (using AES-GCM-128 in all experiments). 
All MPI code was compiled 
with the standard set of MVAPICH compilation
flags and optimization level O2, 
together with the \texttt{-fopenmp} flag to enable OpenMP. 
We compiled BoringSSL 
separately under default settings, 
and linked it with MPI libraries during the linking phase of 
MVAPICH.

Multi-threading is handled via OpenMP-3.1. 
In MVAPICH, we set the following environment variables 
\begin{eqnarray*}
\texttt{MV2\_ENABLE\_AFFINITY=1}  \\	
\texttt{MV2\_CPU\_BINDING\_POLICY=hybrid}  \\
\texttt{MV2\_HYBRID\_BINDING\_POLICY=spread}
\end{eqnarray*}
to distribute resources evenly among MPI ranks.


\heading{Benchmarks.} We consider the following suites of benchmarks: 

\begin{itemize}

\item Ping-pong: This benchmark measures the uni-directional throughput when 
two processes send back and forth via blocking send and receive. 
In each  measurement, 
the two processes send messages of the given size back and forth  for 10,000 times
for messages under 1MB, and 1,000 times otherwise. 


\item OSU micro-benchmark 5.6.2~\cite{OSUBM}: We used the Multiple Pair Bandwidth Test
benchmark in OSU suite to measure aggregate uni-directional throughput of multiple one-to-one
flows between two nodes. 
Each measurement consists of a 100-round loop. 
In each loop the sender sends 64 non-blocking messages of the given size
to the receiver, and waits for the replies. 


\item 2D, 3D, and 4D stencil kernels: We will describe the 2D kernel;
3D and 4D kernels are similar. 
The five-point 2D stencil kernel is as follows. 
In a loop of 1250 rounds, each process in a $\sqrt{N} \times \sqrt{N}$ grid (here $N$ is number of processes) 
performs some matrix multiplications and then sends messages of  $m$~KB 
to its four neighbors via non-blocking sends, and 
finishes the communication via $\Waitall$. 
The kernels allow the computation load to be tuned to achieve a given
ratio between computation and communication. In the experiments,
we set the computation load so that for the unencrypted conventional MPI, 
the average computational time is about $p\%$ of the average total time. 
where $p \in \{25, 50, 75\}$ on the Noleland cluster and 
$p \in \{30, 60, 80\}$ on the PSC Bridges.
  

\item NAS parallel benchmarks~\cite{Bailey:1991:NPB:2748645.2748648}: 
To measure performance of CryptMPI in applications,
we used the CG, LU, BT, and SP tests in the NAS parallel benchmarks, 
all under class D size. 
Collective functions in the NAS benchmarks are unencrypted for both CryptMPI and $\Naive$. 


\end{itemize}

\noskipheading{Benchmark methodology.}
For each experiment, we ran it at
least 10 times, up to 100 times, 
until the standard deviation was within 5\% of the arithmetic mean.
If after 100 measurements, the standard deviation was still big then
we kept running 
 until the 99\% confidence interval 
was within 1\% of the  mean. 

\heading{Multi-thread encryption throughput.}
Before we get into the details
of the benchmark results, it is instructive to
evaluate multi-threading performance of AES-GCM on a single node, 
as it will give us a better understanding of the performance of CryptMPI. 
For each size, we measured the time for performing 500,000 times
the encryption 
of messages of the given size under different numbers of threads. 
The measurements were stable, as they involved only local computation. 
We therefore ran this experiment at least~5 times until the standard
deviation was within 5\% of the arithmetic mean. 

The average encryption throughput of AES-GCM-128 on the local Noleland
is shown in \figref{fig:enc_dec_gcc_throughput}, 
for various message sizes and numbers of threads.  
AES-GCM has about the same the encryption and decryption speed, 
thus the reported throughput here is also the decryption throughput. 
For convenience, in the analysis,
we use the term \emph{encryption-decryption} throughput for the speed of  encrypting and then decrypting a message; 
this is about half of the encryption throughput. 

The encryption throughput on PSC Bridges is shown
in \figref{fig:xsede-enc}, for
various message sizes and numbers of threads. 
The encryption throughput in Bridges is much lower than that
in Noleland, because machines in the latter are newer. 

\begin{figure}[t]
	\centering
		\includegraphics[width=0.45\textwidth]{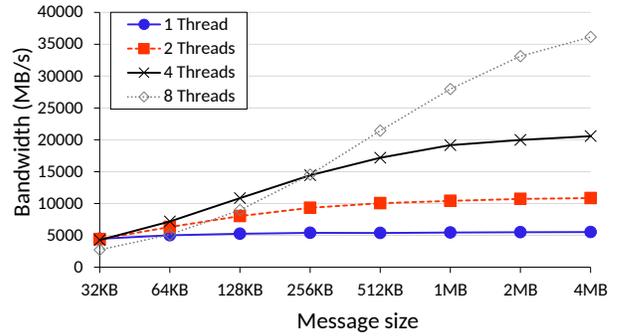}
		\vspace{-1.5ex}
	        \caption{Encryption throughput of AES-GCM (128-bit key) on a Noleland node.}
	\label{fig:enc_dec_gcc_throughput}
	\vspace{-2ex}
\end{figure}

\begin{figure}[t]
	\centering	
		\includegraphics[width=0.45\textwidth]{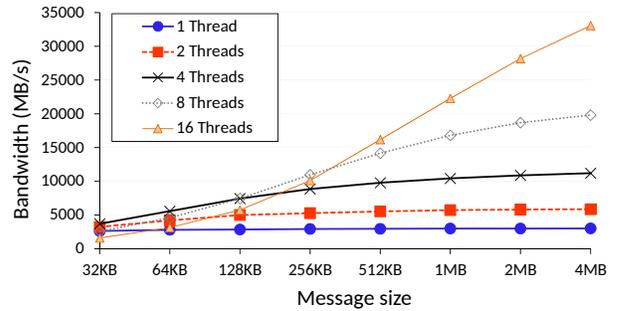}
		\vspace{-1.5ex}
	        \caption{Encryption throughput of AES-GCM (128-bit key) on a PSC Bridges node.}
	\label{fig:xsede-enc}
	\vspace{-2ex}
\end{figure}

\subsection{Results for the Noleland cluster} \label{sec:infiniband}

\noskipheading{Ping-pong.} The ping-pong performance of the 
unencrypted conventional MPI (\Unencrypted), the naive approach (\Naive), 
and CryptMPI is shown in  \figref{fig:pingpong_infiniband}. 
Both {\Naive} and CryptMPI start well below the performance of the
unencrypted baseline. 
For example, for 64KB messages, the overhead of the naive approach, 
compared to the baseline, is  $202.1\%$, 
whereas CryptMPI uses $\min\{T_0-T_1, t\} = 2$ threads
and $k = 1$ chunk, 
ending up with $187.4\%$ overhead. 
As the message size becomes larger, 
the performance of $\Naive$ becomes saturated very quickly,
widening its gap with the baseline. 
In contrast, the performance of CryptMPI increases sharply, 
catching up with the baseline. 
For example, for 4MB messages, the overhead of the $\Naive$ approach is $412.4\%$, 
whereas CryptMPI uses $\min\{T_0-T_1, t\} = 8$ threads and $k = 8$ chunks, 
resulting in just $13.3\%$ overhead.
The phenomenon above can be explained as follows.

\begin{figure}[!tbp]
	\centering
		\includegraphics[width=0.45\textwidth]{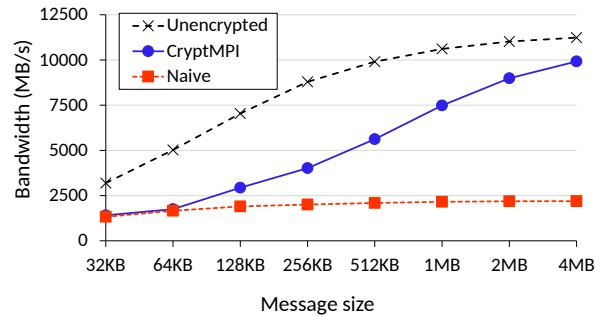}
		\vspace{-1.5ex}
	\caption{Average 
	ping-pong throughput on Noleland.}
	\label{fig:pingpong_infiniband}
	\vspace{-1ex}
\end{figure}

\begin{itemize}
	\item The running time of the naive approach consists of (i)~the (one-thread) encryption-decryption time,
	and (ii) the underlying MPI communication time,
	which is approximately the same as the time of the unencrypted baseline.
	For 4MB messages, the single-thread encryption-decryption throughput (2,759.27 MB/s,  \allowbreak roughly half of the encryption throughput shown
        in \figref{fig:enc_dec_gcc_throughput}) is about
	24\% of the unencrypted ping-pong throughput (11,235.13 MB/s).
	The ping-pong performance of $\Naive$
	can be estimated as
        $\frac{1 + 0.24}{0.24} \approx 5.16$ times slower than
        that of the unencrypted baseline,
	which is consistent with the reported $412.4\%$ overhead above.
	
      \item For messages under 1MB, CryptMPI's running time consists
        of (i) the time of multi-threading
	encryption-decryption, with $t \in \{2, 4, 8\}$ threads,
	and (ii) the underlying MPI communication time.
	For $64$KB messages, $\min\{T_0-T_1, t\} = 2$
	and thus the encryption-decryption throughput (3,151.68 MB/s) is about 62\% of the unencrypted ping-pong throughput  (5,023.81 MB/s).
	The ping-pong performance with CryptMPI can be estimated as
	$\frac{1 + 0.62}{0.62} \approx 2.61$ times slower than
        that of the unencrypted baseline,
	which is close to the reported $187.4\%$ overhead above.
	
      \item For messages of at least 1MB,  the model
        in \secref{sec:oursystem} suggests that 
	the total timing of CryptMPI is roughly  the time of unencrypted Ping Pong 
	plus the encryption-decryption time (under 8 threads) of a \emph{single} 512KB chunk. 
  The 8-thread  encryption-decryption throughput for 512KB (10,736.98 MB/s) is 0.96 times is the unencrypted ping-pong throughput for 4MB (11,235.13 MB/s).
	As a 512KB chunk is only $1/8$ the size of a 4MB message, 
	estimatedly, the ping-pong performance under CryptMPI
	is $\frac{1 + 0.96 \cdot 8}{0.96 \cdot 8} \approx 1.13$ times slower than that of the unencrypted baseline,
	which is consistent with the reported $13.3\%$ overhead above.
\end{itemize}

\noskipheading{OSU Multiple-Pair Bandwidth.}
OSU Multiple-Pair benchmark measures the throughput between two nodes where there are multiple
concurrent flows. 
\figref{fig:multipair_infiniband_64KB} shows the Multiple-Pair
performance of $\Unencrypted$, CryptMPI, and $\Naive$ with 1, 2, 4, 8, and 16 pairs of
communication, and two message sizes, 64KB and 4MB.

When there is only a single pair, we have a similar phenomenon as in
the ping-pong experiment.
If we keep adding more and more pairs, the performance of both CryptMPI and $\Naive$ will increase, 
until the aggregate throughput of all pairs exceeds the bandwidth of the link between the two nodes. 
In that case, the bottleneck is the link,
  and  encryption overhead in this situation is effectively
  hidden. 

In particular, when there are two pairs, $\sysrm$ can match the performance of the baseline,
but $\Naive$ still falls far behind.
For example, for 4MB message,
compared to the baseline, $\sysrm$ has $0.31\%$ overhead,
whereas $\Naive$ has $34.87\%$ overhead.
For four or more pairs, even the naive approach
has about the same performance as the baseline.
This concurs with the observation of Naser et al. that in practical
scenarios where there are multiple
concurrent flows, the overhead of encrypted MPI may be
negligible \cite{Cluster:Naser19}.
Notice that $\Naive$ cannot fully utilize the network bandwidth when
there is one or two pairs regardless of the message size; 
in contrast CryptMPI achieves similar performance as $\Unencrypted$ even with
just a single pair of communication when the message size is sufficiently large. 

\begin{figure}[!htbp]	
\centering
\subfloat[$64$KB-messages]{
  \includegraphics[width=0.4\textwidth]{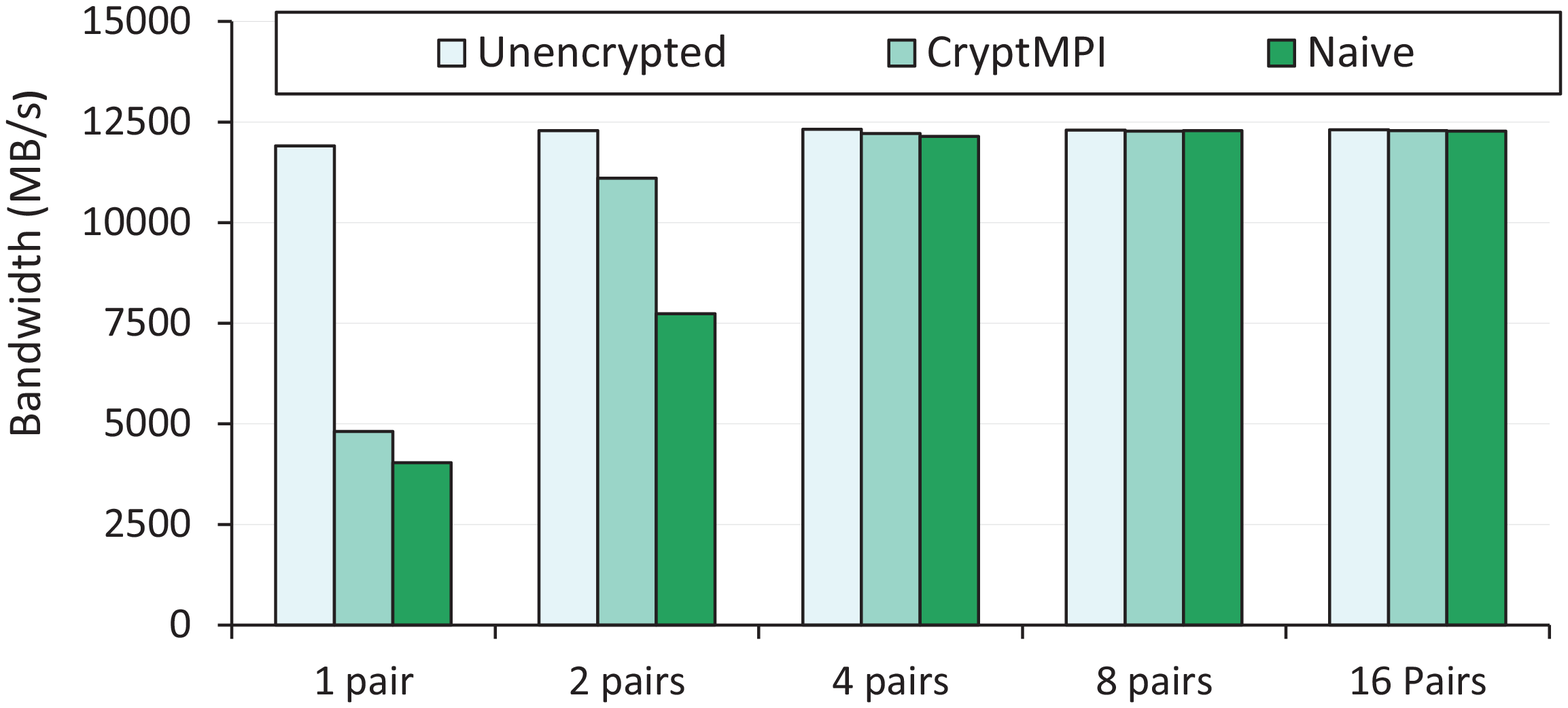}%
}
\\
\subfloat[$4$MB-messages]{
  \includegraphics[width=0.4\textwidth]{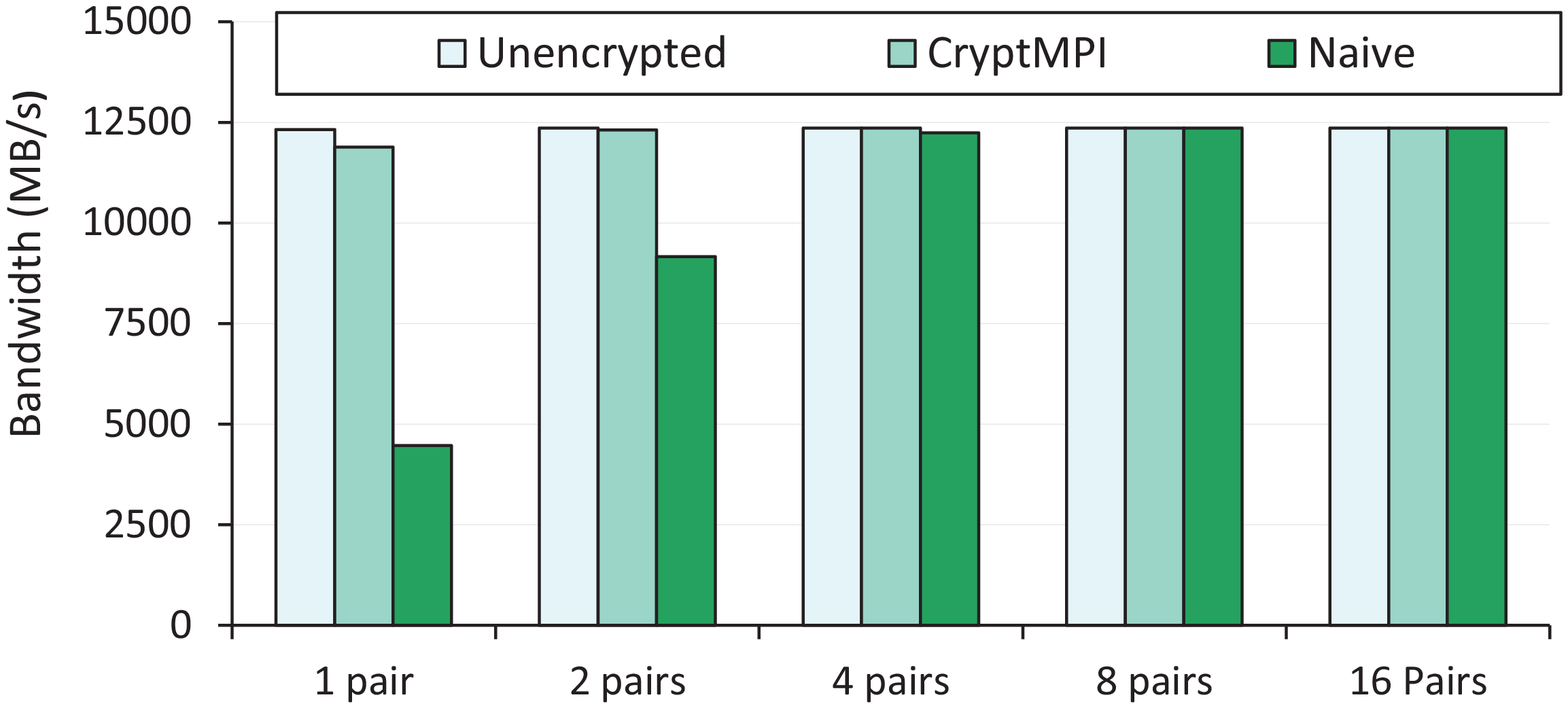}%
}
\vspace{-1.5ex}
\captionsetup{singlelinecheck=false}
\caption{OSU Multiple-Pair throughput on Noleland. }
\label{fig:multipair_infiniband_64KB}
\vspace{-1.5ex}
\end{figure}

CryptMPI is carefully designed to avoid performance degradation. 
For example, consider the setting of two pairs of communication for 4MB messages, 
namely $t = 8$.  
Recall that in the OSU Multiple-Pair Bandwidth test, in each iteration, 
a sender will send 64 messages. 
Initially, CryptMPI  picks $\min\{T_0 - T_1, t\} = 8$ threads and $k = 8$ chunks for each message.
However, after the 8th messages, there are already $64$ pending send requests, 
and CryptMPI will reset $k = 1$ for subsequent messages. 
When there are many pairs, CryptMPI automatically adjusts the thread number to avoid an over-subscription of cores. 
For example, if there are $8$ pairs, since there are $32$ hyper-threads in a Noleland node, 
$T_0 = 32 / 8 = 4$, and thus CryptMPI will use $\min\{T_0 - T_1, t\} = 2$ threads. 



\heading{Stencil and NAS benchmarks.} Results for Stencil and NAS benchmarks on Noleland have similar trends as
those on PSC Bridges. We omit them due to a lack of space. 

\subsection{Results for PSC Bridges} \label{sec:stencil}

\noskipheading{Ping Pong.}
The ping-pong results on Bridges
are shown in \figref{fig:xsede_pingpong}. Because of
the low
of computation power, the naive approach has very high overhead on Bridges.
For instance, on 4MB
naive approach's overhead is 754.87\% compared to the unencrypted base, 
whereas CryptMPI's overhead is just  38.12\%.
We can explain this performance using a similar reasoning as the analysis in Noleland: 

\begin{itemize}
\item For messages under 1MB, CryptMPI's running time consists of (i) the time of multi-threading
    encryption-decryption, with $t \in \{4, 8, 16\}$ threads,
    and (ii) the underlying MPI communication time.
    For $64$KB messages, $\min\{T_0-T_1, t\} = 4$ 
    and thus the encryption-decryption throughput (2,786.08 MB/s) is about 67\% of the unencrypted ping-pong throughput  (4,105.55 MB/s).
    The ping-pong performance under CryptMPI can be estimated as 
    $\frac{1 + 0.67}{0.67} \approx 2.50$ times slower than that of
    the unencrypted baseline,
    which is consistent to the reported $140.2\%$ overhead above.

    \item The 16-thread  encryption-decryption throughput for 512KB (8,091.82 MB/s) is 0.71 times of the unencrypted ping-pong throughput for 4MB (11,404.1 MB/s).
As a 512KB chunk is only $1/8$ the size of a 4MB message,
the ping-pong performance under CryptMPI can be estimated as 
$\frac{1 + 0.71 \cdot 8}{0.71 \cdot 8} \approx 1.18$ times slower than that of the unencrypted baseline,
which is close to the reported $38.12\%$ overhead above.
\end{itemize}

\begin{figure}[!tp]
	\centering
	\vspace{-2ex}
		\includegraphics[width=0.4\textwidth]{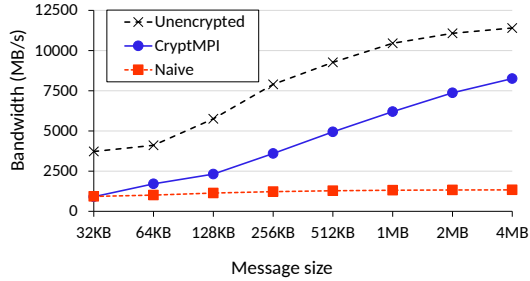}
		\vspace{-1.5ex}
	\caption{Average 
	ping-pong throughput on PSC Bridges.}
	\label{fig:xsede_pingpong}
\end{figure}

\heading{OSU Multiple-Pair bandwidth.}
The multi-pair results on  PSC Bridges 
are shown in
\figref{fig:xsede_multipair_infiniband}.
The trend is similar to
its counterpart in the Noleland cluster with InfiniBand (namely \figref{fig:multipair_infiniband_64KB}):
all libraries have similar performance when there are 
enough pairs of communication.
When the number of concurrent pairs is small,
even for large messages, $\Naive$ has very high overheads. For instance, 
for two pairs and 4MB messages, the overhead for $\Naive$  is 178.48\%, 
whereas CryptMPI's overhead is 5.03\%. 

\begin{figure}[t]	
	\centering
	\vspace{-2.5ex}
	\subfloat[$64$KB-messages]{
	  \includegraphics[width=0.4\textwidth]{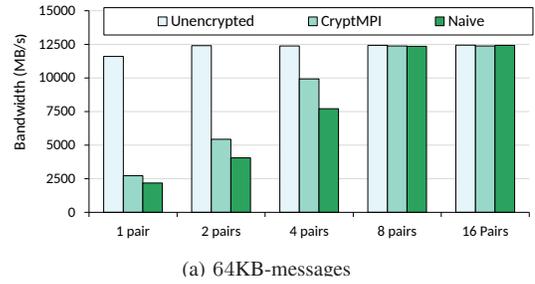}%
	}
	\\
	\vspace{-2.75ex}
	\subfloat[$4$MB-messages]{
	  \includegraphics[width=0.4\textwidth]{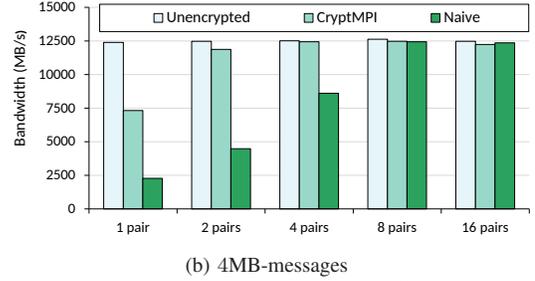}%
	}
	\vspace{-1.5ex}
	\captionsetup{singlelinecheck=false}
	\caption{OSU Multiple-Pair throughput on PSC Bridges. }
	\label{fig:xsede_multipair_infiniband}
	\vspace{-1.5ex}
	\end{figure}

\begin{table*}[!tbp]
		
		\centering
		\caption{Average inter-node communication time ($\Ti$), total communication time ($\Tt$),  and total execution time ($\Te$) in seconds, for NAS Parallel Benchmarks, Class D, 784-rank and 112-node on PSC Bridges. For CG, we used 512-rank and 128-node.}
		\begin{tabular}{p{0.13\linewidth}*{4}{|p{0.05\linewidth}p{0.05\linewidth}p{0.05\linewidth}}}
		\toprule[1.25pt]
		& \multicolumn{3}{c|}{\textbf{CG}} & \multicolumn{3}{c|}{\textbf{LU}} & \multicolumn{3}{c|}{\textbf{SP}} & \multicolumn{3}{c}{\textbf{BT}}  \\
		& $\Ti$ & $\Tt$ & $\Te$ & $\Ti$ & $\Tt$ & $\Te$ & $\Ti$ & $\Tt$ & $\Te$ & $\Ti$ & $\Tt$ & $\Te$ \\ \midrule
		\textbf{Unencrypted} & 7.77 &  13.17 & 27.09 & 6.30  & 32.11 & 52.17 & 24.67 & 35.94 & 62.40 & 26.49 & 40.54 & 69.89 \\
		\textbf{CryptMPI}  & 8.73 &  14.80 & 32.56 & 8.80  & 34.26 & 56.24 & 30.01 & 39.54 & 67.64 & 28.54 & 42.25 & 73.06 \\
		\textbf{Naive}  & 13.91 &  20.93 & 37.84 & 9.71  & 36.06 & 57.87 & 31.62 & 40.11 & 68.18 & 29.03 & 42.67 & 73.49 \\ 
		\bottomrule[1.25pt]
		\end{tabular}
		\label{tab:NAS_PSC_INTER_NODE}
		
\end{table*}

\heading{Stencil Kernels.}
We ran 784-rank Stencil kernels on 112 nodes. 
The results for 2D stencil with different computation loads
are illustrated in~\figref{fig:2d_stencil}. 
When the computational load is not heavy then CryptMPI
significantly improves the performance of the naive  approach. 
For example, for 60\% load and 2MB messages, 
the encryption overhead in CryptMPI is 206\%, 
whereas that of the naive approach is 331\%. 
In this setting CryptMPI uses $\min\{T_0-T_1, t\} = 2$ threads and $k = 4$ chunks. 
Even when computational load is heavy---an unfavorable condition for CryptMPI, 
the improvement of CryptMPI over the naive approach is still noticeable. 
For example, under 80\% computational load and 256KB messages, CryptMPI's encryption overhead is 384\%, 
whereas that of the naive approach is 450\%. 
In this setting CryptMPI instead uses $\min\{T_0-T_1, t\} = 2$ threads and $k = 1$ chunk.



\begin{figure}[!tbp]	
	\centering
	\vspace{-1.5ex}
	\subfloat[$256$KB-messages]{
	  \includegraphics[width=0.4\textwidth]{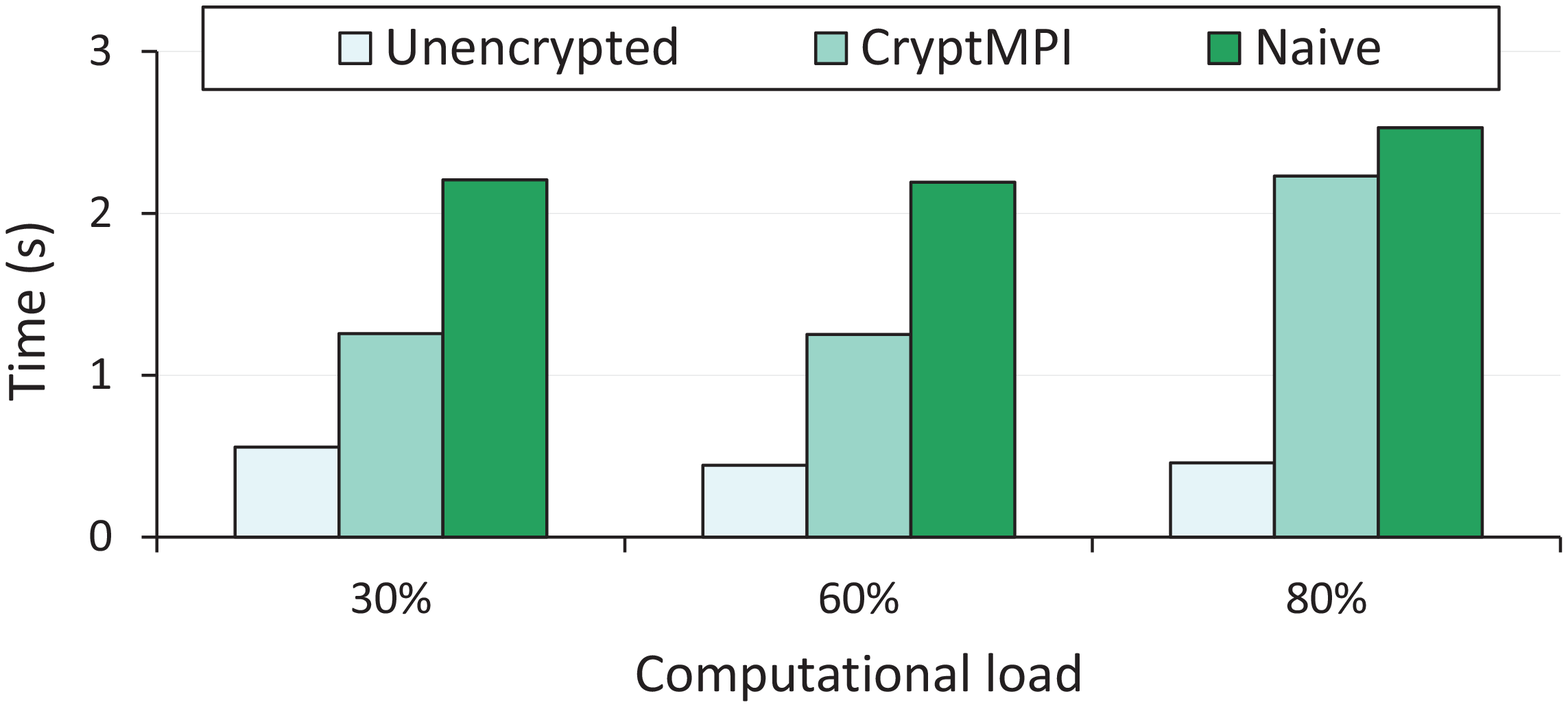}%
	}
    \vspace{-2.5ex}
	\subfloat[$2$MB-messages]{
	  \includegraphics[width=0.4\textwidth]{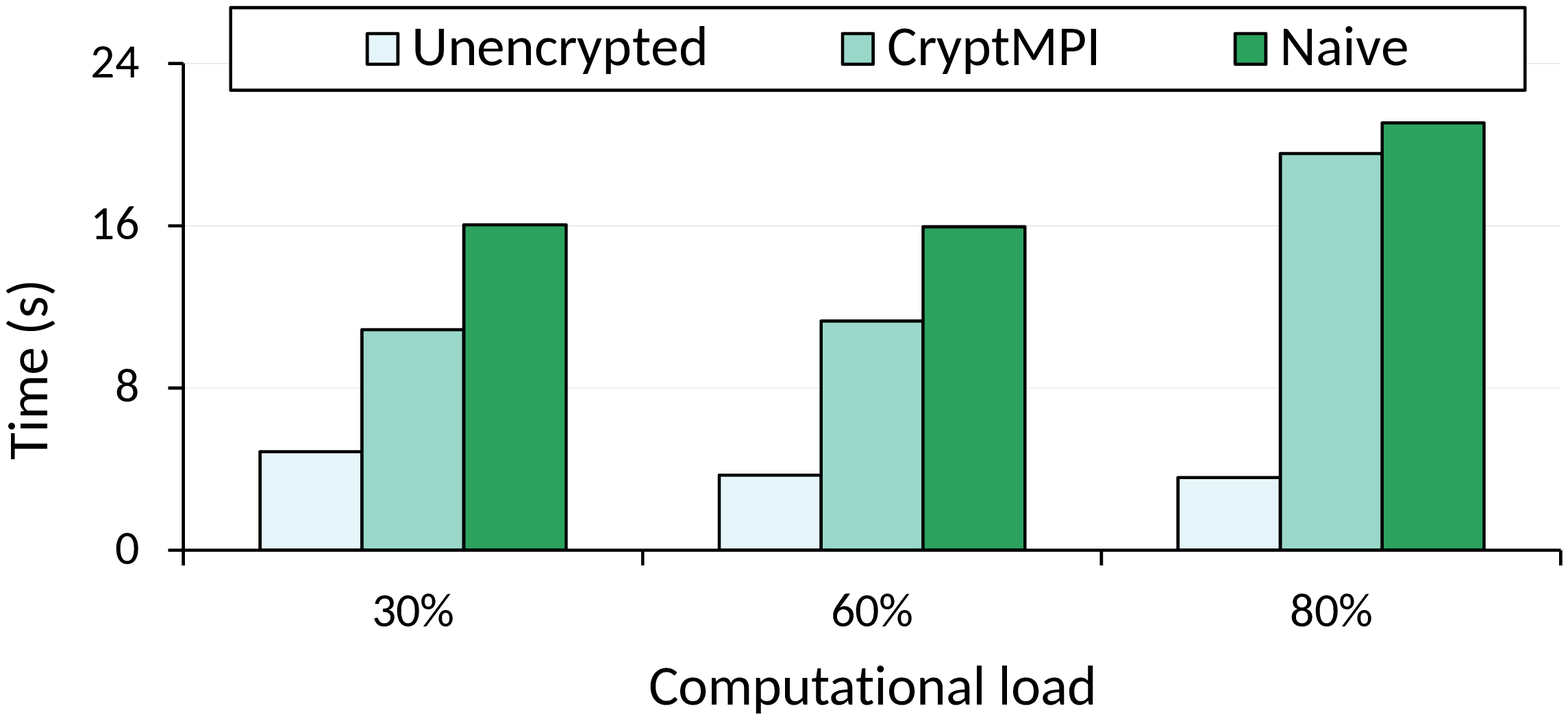}%
	}
	\vspace{-1.5ex}
	\captionsetup{singlelinecheck=false}
	\caption{2D Stencil communication time, for 784-rank and 112-node, PSC Bridges.}
	\label{fig:2d_stencil}
	\vspace{-1.5ex}
	\end{figure}

\heading{NAS benchmarks.} 
The running time of NAS benchmarks on PSC Bridges is shown in Table~\ref{tab:NAS_PSC_INTER_NODE}.
We ran CG with 512 ranks and 128 nodes because CG can only run in the setting where the number of ranks is a power of two. 
The other programs were run with 784 ranks and 112 nodes. The total execution time
overhead of CG is $20.2\%$ and $39.69\%$ for $\sysrm$ and $\Naive$ respectively. 
However, if we take into account only the inter-node communication time
then $\sysrm$'s overhead is only $12.3\%$, 
whereas $\Naive$'s overhead is $79\%$. In BT, 
both $\sysrm$ and $\Naive$ have relatively low overhead for total execution time: $4.53\%$ for CryptMPI 
and $5.15\%$ for $\Naive$. 
The reason for this phenomenon is as follows: 

\begin{itemize}
\item Recall that CryptMPI only improves the naive approach
  for messages of at least 64KB, and thus there might be no noticeable
  difference between the two libraries
  for applications that have largely small messages. 

\item CryptMPI's approach of using multi-thread encryption is only
  beneficial if there are still some idle computational resources to
  exploit, which may not hold for some computationally intensive situations.

\item Finally, CryptMPI's pipelining will give a diminishing return if
  there is already a significant overlapping between computation and
  communication in the application. Such a  situation 
	hides the communication overheads,  and thus both $\Naive$ and CryptMPI should have 
  low overheads. This happens in the BT benchmark. 
\end{itemize}

\section{Conclusion}
\label{sec:conclu}

Achieving high performance for encrypted MPI communication is challenging
since the encryption and decryption operations are the bottlenecks of
encrypted communication in modern platforms. In this work, we present
$\sysrm$, a fast encrypted MPI library. By incorporating
various optimization techniques for large messages,
$\sysrm$ significantly improves
the performance of the naive approach to support MPI communication
with privacy and integrity. The evaluation on two platforms indicates
that the proposed techniques are more effective when the
network speed is significantly higher than single-core encryption
performance, a likely scenario in the future.

A limitation of this work is to consider only one-to-one communication. 
We are currently developing techniques to deal with collective operations
in $\sysrm$.

\bibliographystyle{unsrt}
\bibliography{./crypt}

\end{document}